\pdfoutput=1
\documentclass[a4paper]{article}

\usepackage{microtype}%if unwanted, comment out or use option "draft"

\usepackage{enumerate}
\usepackage{verbatim}
\usepackage{authblk}
\usepackage{amsmath, dsfont}
\usepackage{amsmath}
\usepackage{amssymb}
\usepackage{amsthm}
\usepackage{wasysym}
\usepackage{graphicx}
\usepackage{thmtools}
\usepackage{thm-restate}
\usepackage[boxed]{algorithm2e}
\usepackage{hyperref}

\newcommand{\etal}{\textit{et al.}\ }
\newcommand{\nd}{\mathcal{N}}
\newcommand{\psiu}{\psi_\mathbb{U}}

\newtheorem{lemma}{Lemma}

\newtheorem{theorem}[lemma]{Theorem}

\newtheorem{claim}[lemma]{Claim}

\newcommand{\AT}[1]{\marginpar{\parbox{3cm}{{\small {\bf AT:} #1}}}} %Antonis
\title{The Niceness of Unique Sink Orientations}
\author[1]{Bernd G\"artner \thanks{gaertner@inf.ethz.ch}}
\author[1]{Antonis Thomas\thanks{athomas@inf.ethz.ch}}
\affil[1]{Department of Computer Science\\
	Institute of Theoretical Computer Science, ETH Z\"urich \\
	8092 Z\"urich, Switzerland}

\date{}

\begin{document}

\maketitle

\begin{abstract}
Random Edge is the most natural randomized pivot rule for the simplex algorithm.
Considerable progress has been made recently towards fully understanding its behavior.
Back in 2001, Welzl introduced the concepts of \emph{reachmaps} and \emph{niceness}
of Unique Sink Orientations (USO), in an effort to better understand the behavior of Random Edge.
In this paper, we initiate the systematic study of these concepts. We settle the questions
that were asked by Welzl about the niceness of (acyclic) USO. 
Niceness implies natural upper bounds for Random Edge and 
we provide evidence that these are tight or almost tight in many interesting cases.
Moreover, we show that Random Edge is polynomial on at least $n^{\Omega(2^n)}$ many 
(possibly cyclic)
USO. As a bonus, we describe a derandomization of Random Edge which achieves 
the same asymptotic upper bounds with respect to niceness
and discuss some algorithmic properties of the reachmap.

 \end{abstract}

\section{Introduction}
One of the most prominent open questions in the theory of optimization is whether
linear programs can be solved in strongly polynomial time. In
particular, it is open whether there exists a pivot rule for the
simplex method whose number of steps can be bounded by a polynomial
function of the number of variables and constraints. For most
deterministic pivot rules discussed in the literature, exponential
lower bounds are known. The first such bound was established for
Dantzig's rule by Klee and Minty in their seminal 1972
paper~\cite{kleeminty}; this triggered a number of similar results for
many other rules; only in 2011, Friedmann solved a longstanding
open problem by giving a superpolynomial lower bound for Zadeh's
rule~\cite{Friedmann11}. %\footnote{For this result, Friedmann
%  collected the \$1,000 prize offered by Zadeh in 1980; his PhD thesis
%  received the Tucker Prize in 2012.}

On the other hand, there exists a \emph{randomized} pivot rule,
called \emph{Random Facet}, with an
expected \emph{subexponential} number of steps in the worst case. This
bound was found independently by Kalai~\cite{kalai92} as well as
Matou\v{s}ek, Sharir and Welzl~\cite{msw96} in 1992. Interestingly,
the proofs employ only a small number of combinatorial properties of
linear programs. As a consequence, the subexponential upper bound for
the Random Facet pivot rule holds in a much more general 
abstract setting that encompasses many other (geometric) optimization
problems for which strongly polynomial algorithms are still
missing~\cite{msw96}.

This result sparked a lot of interest in abstract optimization
frameworks that generalize linear programming. The most studied such
framework, over the last 15 years, is that of \emph{unique sink
  orientations} (USO). First described by Stickney and Watson already
in 1978 as abstract models for P-matrix linear complementarity
problems (PLCPs)~\cite{StiWat}, USO were revived by Szab\'{o} and
Welzl in 2001~\cite{SW}. Subsequently, their structural and
algorithmic properties were studied extensively 
(\cite{ss04},\cite{ss05},\cite{matousekcount},\cite{GS},\cite{count},\cite{aoshima},\cite{hpz14},\cite{gt15},\cite{klaus},\cite{hz16}). 
In a nutshell, a USO is an orientation of the $n$-dimensional
hypercube graph, with the property that there is a unique sink in
every subgraph induced by a nonempty face. The algorithmic problem
associated to a USO is that of finding the unique global sink, in an
oracle model that allows us to query any given vertex for the
orientations of its incident edges.

In recent years, USO have in particular been looked at in connection
with another randomized pivot rule, namely \emph{Random Edge} (RE for short). This is 
arguably the most natural randomized pivot rule for the simplex
method, and it has an obvious interpretation also on USO: at every vertex pick 
an edge uniformly at random from the set of outgoing edges and let the other endpoint
of this edge be the next vertex. The path formed constitutes a \emph{random walk}.  
Ever since
the subexponential bound for Random Facet was proved in 1992,
researchers have tried to understand the performance of Random
Edge. This turned out to be very difficult, though. Unlike Random
Facet, the Random Edge algorithm is non-recursive, and tools for a
successful analysis were simply missing. A superexponential lower 
bound on cyclic USO was shown by Morris in
2002~\cite{morris}, but there was still hope that Random Edge might be
much faster on acyclic USO (AUSO).

Only in 2006, a superpolynomial and subexponential \emph{lower} bound for Random Edge on
AUSO was found by Matou\v{s}ek and Szab\'{o}~\cite{ms06}
and, very recently, pushed further by Hansen and Zwick \cite{hz16}. While these
are not lower bounds for actual linear programs, the results
demonstrate the usefulness of the USO framework: it is now clear
that the known combinatorial properties of linear programming are not
enough to show that Random Edge is fast. Note that, in 2011, Friedmann, Hansen
and Zwick proved a subexponential lower bound for Random
Edge on actual linear programs, ``killing'' yet another candidate for
a polynomial-time pivot rule~\cite{fhz11}. 

Still, the question remains open whether Random Edge also has a
subexponential \emph{upper} bound. As there already is a
subexponential algorithm, a positive answer would not be an
algorithmic breakthrough; however, as Random Edge is notoriously
difficult to analyze, it might be a breakthrough in terms of novel
techniques for analyzing this and other randomized algorithms. The
currently best upper bound on AUSO is an exponential improvement over
the previous (almost trivial) upper bounds, but the bound is still
exponential, $1.8^n$ \cite{hpz14}.

%%\paragraph*{Our contribution} heading removed to save space

In this paper, we initiate the systematic study of concepts that are
tailored to Random Edge on USO (not necessarily only AUSO). These
concepts --- \emph{reachmaps} and \emph{niceness} of USO --- were
introduced by Welzl \cite{emonice}, in a 2001 workshop as an
interesting research direction. At that time, it seemed more promising to work on algorithms other than Random Edge; hence, this research direction remained unexplored and the problems posed by Welzl remained open. 
Now that the understanding of Random Edge on USO has advanced a
lot %(in particular with respect to upper bounds), 
we hope that these ``old'' concepts will finally prove 
useful, probably in connection with other techniques.

The reachmap of a vertex is the set of all the coordinates
it can reach with a directed path, and a USO is $i$-nice if for every
vertex there is a directed path of length at most $i$ to another
vertex with smaller reachmap. Welzl pointed out that the concept of
niceness provides a natural upper bound for the Random Edge algorithm.
Furthermore, he asks the following question: ``Clearly every unique
sink orientation of dimension $n$ is $n$-nice.  Can we do better? In
particular what is the general niceness of acyclic unique sink
orientations?''

We settle these questions, in Section~\ref{sec:bounds}, by proving that for AUSO $(n-2)$-nice is tight, meaning that $(n-2)$ is an upper bound on the 
niceness of all AUSO and there are AUSO that are \emph{not} $(n-3)$-nice. For cyclic USO we argue that $n$-nice is tight.
In Section~\ref{sec:reachmap}, we give the relevant definitions and in Section~\ref{sec:REnice} we show an upper bound of $O(n^{i+1})$ for the number of steps 
RE takes on an $i$-nice USO. In addition, we describe a derandomization of RE which also
needs at most $O(n^{i+1})$ steps on an $i$-nice USO, thus matching the behavior of RE. 

Furthermore, we provide the following observations and results as applications for the concept of niceness.
In Section~\ref{sec:known}, we argue that RE can solve the AUSO instances that have been designed
  as lower bounds for other algorithms 
  (e.g. Random Facet \cite{matouseklbrandomfacet, G02} or  Bottom Antipodal \cite{ss05}) in polynomial time.
  In addition, we prove in Section~\ref{sec:count} that RE needs at most a
  quadratic number of  steps in at least $n^{\Theta(2^n)}$ many, possibly cyclic, USO. 
  The previous largest class of USO on which RE is polynomial (quadratic) is that 
  of decomposable USO; we include a proof that the number of those is $2^{\Theta(2^n)}$ and, thus, our new 
  result is a strict improvement. %We provide the 
  
  Finally, we provide an application for the concept of reachmap. 
  In Section~\ref{sec:FS}, we describe a new algorithm that is a variant of the Fibonacci Seesaw algorithm
  (originally by Szab\'o and Welzl \cite{SW}). The number of vertex evaluations it needs to solve a USO can 
  be bounded by a function that is exponential to the size of the reachmap of the starting vertex. 
  In contrast, the Fibonacci Seesaw needs a number of vertex evaluations that is exponential to 
  the dimension of the USO.

\section{Preliminaries} \label{sec:reachmap}
We use the notation $[n] = \{1,\ldots n\}$. 
Let $Q^n = 2^{[n]}$ be the set of vertices of the $n$-dimensional hypercube. 
A vertex of the hypercube $v \in Q^n$ is denoted by the set of coordinates 
it contains.  The symmetric difference of two vertices, denoted as 
$v \oplus u$ is the set of coordinates in which they differ. 
%Now, let $A \subseteq [n]$. With $Q^A$ we mean the cube over the coordinates defined in $A$;
%a $|A|$-dimensional cube. In this notation $Q^n$ could also be written as $Q^{[n]}$ but we prefer the former. 
%With $\bar{A}$ we denote the set $[n] \setminus A$, except when explicitly defined otherwise.  \AT{$\bar{A}$ notation may not be needed.}
%Similarly, given a vertex $u \in Q^n$, with $\bar{u}$ we mean the antipodal vertex; that is vertex $u'$ 
%such that $u \oplus u' = [n]$.  
Now, let $J \in 2^{[n]}$ and $v\in Q^n$. A \emph{face} of the hypercube, $F_{J,v}$, is defined as the set of vertices 
that are reached from $v$ over the coordinates defined by any subset of $J$,
i.e. $F_{J,v} = \{u\in Q^n | v\oplus u \subseteq  J\}$. The dimension of the face is $|J|$. 
We call edges the faces of dimension 1, e.g. $F_{\{j\},v}$, and vertices the faces of dimension 0. 
The faces of dimension $n-1$ are called facets. For $k\leq n$ we call a face of dimension $k$ a $k$-face.

Let $v,u\in Q^n$. 
By $|v\oplus u|$ we denote the Hamming distance (size of the symmetric difference) of $v$ and $u$. 
Given $v \in Q^n$, we define the neighborhood of $v$ as $\nd(v) = \{u\in Q^n | \text{  } |v\oplus u| = 1\}$.
Now,
let $\psi$ be an orientation of the edges of the $n$-dimensional hypercube. %\AT{Write what is a USO.}
Let $v,u\in Q^n$. The notation $v \xrightarrow{j} u$ (w.r.t $\psi$) means that $F_{\{j\},v} = \{v,u\}$ and 
that the corresponding edge is oriented from $v$ to $u$ in $\psi$.
Sometimes we write $v \rightarrow u$, when when the coordinate is irrelevant.
An edge $v \xrightarrow{j} u$ is forward if $j \in u$ and otherwise we say it is backward.

We say that $\psi$ is a \emph{Unique Sink Orientation} (USO) if every non-empty face 
has a unique sink.
In the rest we write $n$-USO to mean a USO over $Q^n$.
Here $n$ is always used to mean the dimension of the corresponding USO.
Consider a USO $\psi$; we define its outmap $s_\psi$, in the spirit of Szab\'o
 and Welzl \cite{SW}.
The \emph{outmap} is a function $s_\psi: Q^n \rightarrow 2^{[n]}$, defined by
$ s_\psi(v) = \{ j \in [n] | v \xrightarrow{j} v \oplus \{j\} \}$ %\exists u\in \nd(v) \text{ s.t. }  v \xrightarrow{j} u \text{ w.r.t } \psi \}$
for every $v\in Q^n$. 
A sink of a face $F_{J,v}$ is a vertex $u \in F_{J,v}$, such that $s_\psi(u) \cap J = \emptyset$.
We mention the following lemma w.r.t. the outmap function.
\begin{lemma}[\cite{SW}] \label{lem:outmap}
For every USO $\psi$, $s_\psi$ is a bijection.
\end{lemma}
The algorithmic problem for a USO $\psi$ is to find the global sink, i.e. find $t \in Q^n$ such that $s_\psi(t) = \emptyset$.
The computations take place in the \emph{vertex oracle} model: We have an oracle that given a vertex $v \in Q^n$, 
returns $s_\psi(v)$ (vertex evaluation). This is the standard computational model in the USO literature and all the upper and lower 
bounds refer to it. %are with respect to this model.

\subparagraph*{Reachmap and niceness.} We are now ready to define the central concepts of this paper.
 Given vertices $v,u\in Q^n$ we write $v \rightsquigarrow u$ if  
there exists a directed path from $v$ to $u$ (in $\psi$). We use $d(v,u)$ to denote the length of the shortest directed
path from $v$ to $u$; if there is no such path then we have $d(v,u) = \infty$ and otherwise we have $d(v,u) \geq |v \oplus u|$.
The following lemma is well-known and easy to prove by induction on $|v \oplus u|$.
\begin{lemma} \label{lem:pathtosink}
For every USO $\psi$, let $F \subseteq Q^n$ be a face and $u$ the sink of this face. Then, for every 
vertex $v \in F$ we have $d(v,u) = |v \oplus u|$.
\end{lemma}
Subsequently, we define the \emph{reachmap} $r_\psi:Q^n \rightarrow 2^{[n]}$, for every $v\in Q^n$, as:
%\[r_\psi(v) = s_\psi(v) \cup \bigcup \{r_\psi(u) | v \rightsquigarrow u \text{ w.r.t } \psi \}.\]
\[r_\psi(v) = s_\psi(v)  \cup \{j \in [n] | \exists u\in Q^n \text{ s.t. } v \rightsquigarrow u \text{ and } j \in s_\psi(u) \}.\]
Intuitively, the reachmap of a vertex contains all the coordinates that the vertex can reach with 
a directed path.
We say that vertex $v \in Q^n$ is $i$-covered by vertex $u\in Q^n$,
if $d(v,u) \leq i$ and  $r_\psi(u)\subset r_\psi(v)$ (proper inclusion).
Then, we say that a USO $\psi$ is \emph{$i$-nice} if every vertex $v\in Q^n$  (except the global sink) is $i$-covered by 
some vertex $u\in Q^n$.
Of course, every $n$-USO $\psi$ is $n$-nice since every vertex $v$ is $n$-covered by the sink $t$.
Moreover,  $r_\psi(v) \supseteq v\oplus t$, for every vertex $v \in Q^n$.

It is not difficult to observe that every USO in 1 or 2 dimensions is 1-nice, but the situation changes in 
3 dimensions. Consider the illustration in the figure below. 

\begin{figure}[htbp] 
	\centering
	\includegraphics[scale=1]{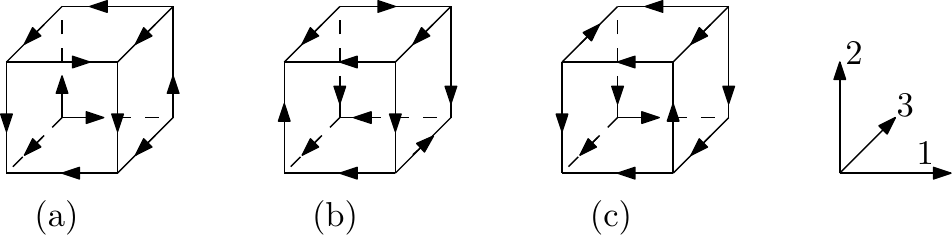}
	\caption{Examples of 3-dimensional USO: (a) Klee-Minty, which is 1-nice. 
	(b) The only 2-nice 3-dimensional AUSO which is not 1-nice. (c) The only cyclic USO in 3 dimensions, which is 3-nice.} 
	\label{fig:cubes}
\end{figure}
Let us note  that the AUSO in Figure~\ref{fig:cubes}b is the largest AUSO which is 
not $(n-2)$-nice. As we prove in Theorem~\ref{thm:UBnice}, every $n$-AUSO with $n\geq 4$ is 
$(n-2)$-nice.

\subparagraph*{Algorithmic properties of the reachmap.}  
%\AT{Is this section keeper or not? Probably we'll have to remove it due to space constraints in which case I would write couple of paragraphs about it?}
Our focus lies mostly on the concept of niceness. Nevertheless, we briefly discuss some of the algorithmic properties of the 
reachmap here. 

It was proved by the authors, in \cite{gt15}, that when given an AUSO $\psi$ described succinctly by a Boolean circuit, and two vertices $s$ and $t$, 
deciding if $s \rightsquigarrow t$ is $PSPACE$-complete. This means that the input is a Boolean circuit, of size polynomial in $n$, with $n$ 
input and $n$ output gates: the input to the circuit is a vertex and the output is the outmap of that vertex according to $\psi$. 

More recently, Fearnley and Savani \cite{fs16} proved %(with completely different techniques) 
that deciding whether the Bottom Antipodal algorithm (this is the algorithm that from a vertex $v$ jumps to vertex $v \oplus s_\psi(v)$), started at 
vertex $v$ will ever encounter a vertex $v'$ such that $j \in s_\psi(v')$, for a given coordinate $j$, is $PSPACE$-complete. 
This line 
of work was initiated in \cite{apr14} and further developed in \cite{ds15} and \cite{fs15} and aims at understanding the computational
power of pivot algorithms \cite{fs16}. Below, we provide a related theorem: it is $PSPACE$-complete to decide if a coordinate is in the reachmap
of a given vertex in an AUSO. It is, thus, computationally hard to discover the reachmap of a vertex.
\begin{theorem} \label{thm:pspace}
Let $\psi$ be an $n$-AUSO (described succinctly by a Boolean circuit), $v\in Q^n$ and $j \in [n]$. It is $PSPACE$-complete
to decide whether $j \in r_\psi(v)$.
\end{theorem}
\begin{proof}
We provide a reduction from AUSO-Accessibility to prove $PSPACE$-hardness. 
The $PSPACE$ upper bound follows standard arguments that can be found in \cite{gt15}.
The input consists of $C_\psi$ (the circuit), which represents 
the $n$-AUSO $\psi$, and two vertices $s,t \in Q^n$.
We construct an $(n+1)$-USO $\psi'$ from $\psi$. 

First, all the edges on coordinate $n+1$ are backwards uniform, with one exception that we discuss later. 
We embed the orientation $\psi$ in the faces $A$ and $B$ illustrated in the figure below.
We flip the edge
$F_{\{n+1\}, t}$ which is safe as the outmaps of the two vertices involved 
differ only in the connecting coordinate.
An illustration of the construction appears in Figure~\ref{fig:pspace}.

\begin{figure}[htbp]
	\centering
	\includegraphics[scale=0.75]{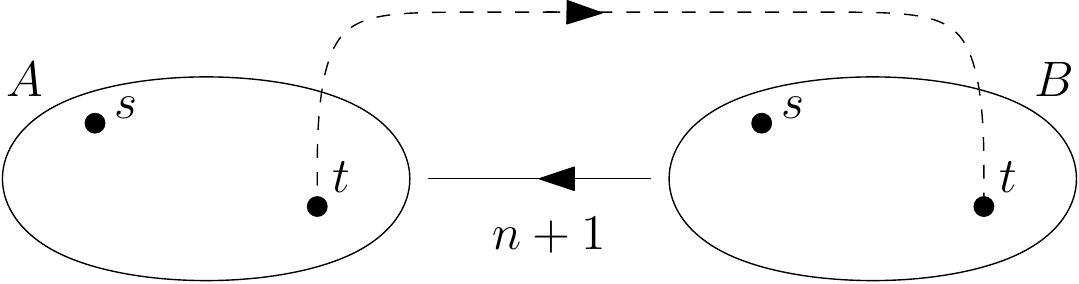}
	\caption{An illustration of the construction. The vertices $s$ and $t$ are the ones from the input of 
	the AUSO-Accessibility instance. The flipped edge appears dashed in the figure.} 
	\label{fig:pspace}
\end{figure}

This defines $\psi'$. Note that $\psi'$ is an AUSO follows from Lemma~\ref{lem:product}.
Consider the vertex $s \in A$. We have that $\{n+1\} \in r_\psi(s)$
if and only if there is a path $s \rightsquigarrow t$. 
\end{proof}

Finally, we want to note that it is natural to upper bound 
algorithms on AUSO by the reachmap of the starting vertex. Any reasonable path-following algorithm that solves 
an AUSO $\psi$ in $c^n$ steps, for some constant $c$, can be bounded by $c^{|r_\psi(s)|}$ where $s$ is the starting vertex. 
The reason is that the algorithm will be contained in the cube $F_{r_\psi(s), s}$ of dimension $|r_\psi(s)|$. 
Moreover, we claim that this is also possible for algorithms that are not path-following. As an example we give 
in Section~\ref{sec:FS}  a variant of the Fibonacci Seesaw algorithm of \cite{SW} that runs in time 
$c^{|r_\psi(s)|}$ for some $c < \phi$ (the golden ratio).

\section{Random Edge on $i$-nice USO} \label{sec:REnice}
In this section we describe how RE behaves on $i$-nice USO. We give a 
natural upper bound and argue that it is tight or almost tight in
many situations. In addition, we give a simple derandomization
of RE, which asymptotically achieves the same upper bound. % w.r.t. niceness.
Firstly, we consider the following natural upper bound.
\begin{theorem} \label{thm:UB}
Started at any vertex of an $i$-nice USO, Random Edge will perform an expected number of at most $O(n^{i+1})$ steps.
\end{theorem}
\begin{proof}
  For every vertex $v$, there is a directed path of length at most $i$
  to a \emph{target} $t(v)$, some fixed vertex of smaller reachmap. At
  every step, we either reduce the distance to the current target (if
  we happen to choose the right edge), or we start over with a new
  vertex and a new target. The expected time it takes to reach
  \emph{some} target vertex can be bounded by the expected time to
  reach state $0$ in the following Markov chain with states
  $0,1,\ldots,i$ (representing distance to the current target): at
  state $k>0$, advance to state $k-1$ with probability $1/n$, and fall
  back to state $i$ with probability $(n-1)/n$. A simple inductive
  proof shows that state $0$ is reached after an expected number of
  $\sum_{k=1}^{i}n^k = O(n^i)$ steps. Hence, after this expected
  number of steps, we reduce the reachmap size, and as we have to do
  this at most $n$ times, the bound follows.
\end{proof} 
Already, we can give some first evidence on the usefulness of niceness for analyzing RE: 
Decomposable orientations 
have been studied extensively in literature. The fact that RE terminates in $O(n^2)$ steps 
on them has been known at least since the work of Williamson-Hoke \cite{wh}. 
Let a coordinate be \emph{combed} if all edges  
on this coordinate are directed the same way. Then, a cube orientation is \emph{decomposable} 
if in every face of the cube there is a combed coordinate. The class of decomposable orientations,
known to be AUSO, contains the Klee-Minty cube \cite{kleeminty} (defined combinatorially in \cite{ss05}). 

It is straightforward to argue that such 
orientations are 1-nice and, thus, our upper bound from Theorem~\ref{thm:UB}
is also quadratic. Moreover, quadratic lower bounds have been proved for the behavior of RE on 
Klee-Minty cubes \cite{bp}. 
This line of work was initiated in \cite{ghz98}, where the lower
bound of $\Omega(n^2 / \log n)$ was proved, and it was finalized in 
\cite{bp}, where the lower bound of $\Omega(n^2)$ was proved. 
We conclude that, for $1$-nice USO, the upper bound in Theorem~\ref{thm:UB} is optimal. 
Let us note that, so far, the largest class of USO known to be solvable in polynomial time (specifically quadratic)
by RE is decomposable. 
In Section~\ref{sec:count} we will prove that the class of 1-nice USO is strictly larger than that of 
decomposable and contains, also, cyclic USO.
%Inspired by this, we investigate further the class of $1$-nice USO in Section~\ref{sec:1nice}.

\subsection{A derandomization of Random Edge}
Consider the \emph{join} operation. Given two vertices $u,v$, $join(u,v)$ is a vertex $w$
such that $u \rightsquigarrow w$ and $v \rightsquigarrow w$. We can compute $join(u,v)$ 
as follows: by Lemma~\ref{lem:outmap}, there must be a coordinate, say $j$, such that $j \in s_\psi(u) \oplus s_\psi(v)$.
Assume, w.l.o.g., that $j \in s_\psi(u)$. Consider the neighbor $u'$ of $u$ such that $u \xrightarrow{j} u'$. 
Recursively compute $join(u',v)$.  %\AT{Why does this take linear steps?}
It can be seen by induction on $|u \oplus v|$ that the $join$ operation takes $O(n)$ time.
Similarly, we talk about a join of a set $S$ of vertices. A $join(S)$ is a vertex $w$ such that 
ever vertex in $S$ has a path to it. We can compute $join(S)$ by iteratively joining all the vertices in $S$.
%(all this is with slight abuse of notation, as the join of a pair ---respectively set of vertices--- is not uniquely defined).

Furthermore, let $\nd^+(v) = \{u\in \nd(v) | v \rightarrow u\}$ denote the set of out-neighbors 
of a vertex $v$. In the subsequent lemma, we argue that the vertices in $\nd^+(v)$ can be joined
with linearly many vertex evaluations.  

\begin{lemma} \label{lem:1nicejoin}
Let $\psi$ be an $n$-USO and  $v\in Q^n$ a vertex with $\nd^+(v)$ already known. Then, there is an algorithm that joins the vertices in $\nd^+(v)$
with $|s_\psi(v)|$ many vertex evaluations. 
\end{lemma}
\begin{proof}
First, we evaluate all the vertices in $\nd^+(v)$. We maintain a set of active vertices $AV$ and a set of active coordinates $AC$. 
Initialize $AV = \nd^+(v)$ and $AC = s_\psi(v)$. The algorithm keeps the following invariants: every vertex that gets 
removed from $AV$ has a path to some vertex in $AV$; also for every vertex $u$ s.t. $v\xrightarrow{l} u$, $u \in AV$ if and only if $l \in AC$. 

Then, for each $u \in AV$: for each $l \in AC$: %If there is coordinate $l \in AC$ such that $l \notin s_\psi(u)$ and $\{l\} \neq (u \oplus v)$ then we update
if $l \notin s_\psi(u)$ and $\{l\} \neq (u \oplus v)$ then we update
$AC \leftarrow AC \setminus \{l\}$ and $AV \leftarrow AV \setminus (v\oplus \{l\})$. See Figure~\ref{fig:neighborjoin}.
			\begin{figure}[htbp]
				\centering
				\includegraphics[scale=0.75]{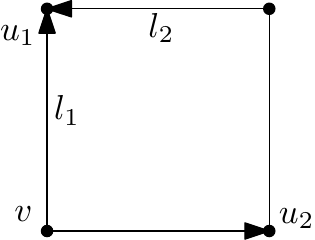}
				\caption{We have $u,w \in AV$, $l \in AC$,  $l \notin s_\psi(u)$ and $\{l\} \neq (v \oplus u)$. Thus, 
				the edge $F_{j, w}$ has to be outgoing for $w$. Hence, $w \rightsquigarrow u$ and the algorithm removes 
				$w$ from $AV$ and $l$ from $AC$.}
				\label{fig:neighborjoin}
			\end{figure}
			
If in the above loop the vertex $u$ is the sink of the face $F_{AC, u}$ then terminate and return $v'=u$. Of course, in this case
every vertex in $AV$ has a path to $u$. Otherwise the loop will terminate when there is no coordinate in $AC$ that satisfies the conditions above.
In this case 
we have that $\forall u \in AV$, $u$ is the source of the face $F_{AC\setminus (u\oplus v),u}$.
That is, it is the source of the face spanned by the vertex and all the active coordinates  $AC$ except the one that connects it to $v$.
In this case, we return the vertex $v'=(v \oplus AC)$. We have that every vertex in $AV$ has a path to $v'$: this is because in any USO 
the source has a path to every vertex (this can be proved similarly to Lemma~\ref{lem:pathtosink}).
\end{proof}

Using Lemma~\ref{lem:1nicejoin}, we can now argue that there exists a derandomization of Random Edge that
asymptotically matches the upper bound of Theorem~\ref{thm:UB}.

\begin{theorem} 
There is a deterministic algorithm that finds the sink of an $i$-nice $n$-USO $\psi$ with $O(n^{i+1})$ vertex evaluations.
\end{theorem}
\begin{proof}
Let $v$ be the current vertex. Consider the set $R_i \subseteq 2^{[n]}$ of vertices that are reachable along
directed paths of length at most $i$ from $v$. Since $\psi$ is $i$-nice, we know that at least one of them has strictly smaller
reachmap. In particular, any vertex reachable from all the vertices in $R_i$ has a smaller reachmap. 
Thus, we compute a $join$ of all the vertices in $R_i$.  

Consider the set $R_{i-1}$. The size of $R_{i-1}$ is bounded by $|R_{i-1}| \leq \sum_{k=0}^{i-1} \binom{n}{k} \leq \sum_{k=0}^{i-1} n^k$ 
and, thus, $|R_{i-1}| = O(n^{i-1})$.
Every vertex in $R_i$ can be reached in one step from some vertex in $R_{i-1}$. Assume that none of the vertices in $R_{i-1}$ is the 
sink; otherwise, the algorithm is finished. Then, for every vertex $v \in R_{i-1}$ we join $\nd^+(v)$ with the algorithm from 
Lemma~\ref{lem:1nicejoin}, with $O(n)$ vertex evaluations. Therefore, with $O(n^i)$ vertex evaluations we have 
a set $S$ of $O(n^{i-1})$ many vertices
and each $v' \in S$ is a join of $\nd^+(v)$ for some vertex $v \in R_{i-1}$. 

The next step is to join all the vertices in set $S$, using the algorithm at the beginning of the current section, which takes $O(n)$ for 
each pair of vertices. Hence, the whole procedure will take an additional $O(n^i)$ vertex evaluations. The result is a vertex $u$ 
that joins all the vertices in $R_i$ and thus $i$-covers $v$. 
Because the size of the reachmap decreases by at least one in each round,
we conclude that this algorithm will take at most
$O(n^{i+1})$ steps.

Finally, note that to achieve this upper bound we do not need to know that the input USO is $i$-nice.
Instead, we can iterate through the different values of $i=1,2,\ldots$ without changing the asymptotic 
behavior of the algorithm.
\end{proof}

\section{On the niceness of known lower bound constructions}  \label{sec:known}
As further motivation for the study of niceness of USO, we want to argue that 
RE can solve the AUSO instances that were designed as lower bounds for 
other algorithms in polynomial time. 
This is because of provable upper bounds on the niceness of 
those constructions. With similar arguments, upper bounds on the niceness 
of the AUSO that serve as subexponential lower bounds for RE can be shown;
thus, RE has upper bounds on these constructions that are almost matching to the
lower bounds.
This can be seen as a direct application of the concept of niceness.

To argue about the upper bounds on the niceness of the known constructions,
we first have to describe the standard tools used to construct USO that were introduced in \cite{ss04}.  

\subsection{The constructive lemmata of \cite{ss04}}  \label{sec:constrlem}
Schurr and Szab\'o, in \cite{ss04}, kickstarted a new direction for lower-bounding 
algorithms on AUSO. Not only because they proved by an adversarial argument that any deterministic algorithm needs
at least $\Omega{(n^2 / \lceil \log n \rceil)}$ steps to solve an AUSO. But, also, because they
provided the methods for constructing USO, which where used in most of the constructions we discuss later. 
Here we rewrite the two constructive lemmas for the sake of completeness and 
in order to argue about the preservation of niceness when they are applied.

\begin{lemma}[Product Lemma] \label{lem:product}
Let $A$ be a set of coordinates, $B\subseteq A$ and $\bar{B} = A \setminus B$.
Let $\tilde{s}$ be a USO on $Q^{B}$ and let $s_u$, $u \in Q^{B}$, be $2^{|B|}$ USOs on $Q^{\bar{B}}$. Then, the orientation defined
by the outmap 
\[ s(v) = \tilde{s}(v \cap B) \cup s_{v \cap B}(v \cap \bar{(B)}) \]
on $Q^A$ is a USO. Furthermore, if $\tilde{s}$ and all $s_u$ are acyclic so is $s$. \vskip0.3cm 
Let $z$ be the sink of $\tilde{s}$. If $\tilde{s}$ is $i$-nice and $s_z$ is $i$-nice then so is $s$.
\end{lemma}
\begin{proof}
The lemma is proved in \cite{ss04}. Here we will only argue about the niceness part. 

Let $\tilde{s}$ be $i$-nice and consider $F_{\bar{B}, z}$, the $|\bar{B}|$-dimensional face that corresponds to $z$.
Consider any vertex $v' \in Q^A \setminus F_{\bar{B}, z}$ and let $u \in Q^B$ be such that $v = v'\cap B$ is $i$-covered by $u$ in 
$\tilde{s}$. Then, we have that $v'$ is $i$-covered by $u' = u \cup (v' \cap \bar{B})$ in $s$. First, there is a path from $v'$ to $u'$ because 
there is such a path in $\tilde{s}$. Thus, we have that $r(u') \subseteq r(v')$. 
By our assumptions, there exists a coordinate $l \in B$ such that $l \in \tilde{r}(v) \setminus \tilde{r}(u)$. It is the case that 
$l \in r(v') \setminus r(u')$. That means that $r(u') \subset r(v')$ and thus $v'$ is $i$-covered by $u'$.

Now, let $v' \in F_{\bar{B}, z}$. Since every vertex in $F_{\bar{B}, z}$ corresponds to the sink $z$ of $\tilde{s}$, it 
cannot be $i$-covered by a vertex outside of $F_{\bar{B}, z}$. Since, we have that $s_z$ is $i$-nice it is the case 
that $v'$ is $i$-covered by some vertex $u' \in F_{\bar{B}, z}$. Note that if $s_z$ is $i'$-nice for $i' > i$ then
$s$ would only be $i'$-nice.
\end{proof}

The first dimension where there are USOs that are not 1-nice is 3. Therefore, for the above lemma and $|\bar{B}|\leq 2$ 
if $\tilde{s}$ is 1-nice then so is $s$. Following, we state the second lemma.

\begin{lemma}[Hypersink Reorientation] \label{lem:hypersink}
Let $A$ be a set of coordinates, $B\subseteq A$ and $\bar{B} = A \setminus B$. Let $s$ be a USO on $Q^A$ and let 
$Q^B$ be a subcube of $Q^A$. If $s(v) \cap \bar{B} = \emptyset$ for all
$v \in Q^B$ and $\tilde{s}$ is a USO on $Q^B$, then the outmap $s'(v) = \tilde{s}(v\cap B)$ for $v \in Q^B$ and 
$s'(v) = s(v)$ otherwise is a USO on $Q^A$.

If $s$ and $\tilde{s}$ are acyclic, then so is $s'$. %Finally, if $s$ and $\tilde{s}$ are $i$-nice, then so is $s'$.
\end{lemma}

Unfortunately, the above lemma does not carry niceness. This is easy to see: Assume that $s$ is $i$-nice. 
There can be a vertex $v$ that is $i$-covered in $s$ by a vertex $u$. However the closest neighbor 
of $u$ in the hypersink $Q^B$ might be the source of subcube $Q^B$ in which case $v$ is not $i$-covered anymore in $s'$.
It is not difficult to construct such examples. 
Also note that the edge flip operation we used in Section~\ref{sec:lbAUSO} is a corollary of this lemma. In the construction 
of Theorem~\ref{thm:lbAUSO} we start with the uniform orientation which is decomposable and thus $1$-nice and end
with a $(n-2)$-nice AUSO. 

\begin{comment}
The upper bound for RE on the cyclic USO of Morris \cite{morris} is asymptotically matching 
the lower bound.
 We summarize the relevant information in the following table 
and describe the details on how to obtain it in Appendix~\ref{app:nicenessknown}.
\begin{center}
\begin{tabular}{l|c|c|c|c} 
Algorithm & Reference & Lower bound & Niceness & RE Upper bound \\ \hline
Random Facet & \cite{matouseklbrandomfacet},\cite{G02} & $ 2^{\Theta(\sqrt{n})}$ & $1$ & $O(n^2)$ \\
%General deterministic& \cite{ss04} & $\Omega(n^2 / \lceil \log n \rceil)$ & $\lceil \log n \rceil$ & $n^{O(\lceil \log n \rceil)}$  \\
Bottom Antipodal & \cite{ss05} & $\Omega(\sqrt{2}^n)$ & 2 & $O(n^3)$ \\
RE acyclic &\cite{ms06} & $2^{\Omega{(n^{1/3})}}$ & $n^{1/3}$ & $2^{O(n^{1/3} \log n)}$ \\
RE acyclic & \cite{hz16} & $2^{\Omega(\sqrt{n \log n})}$ & $\sqrt{n}$ & $2^{O(\sqrt{n} \log n)}$ \\
RE cyclic & \cite{morris} & $\frac{n-1}{2}!$ & $n$ & $n^{O(n)}$  
\end{tabular} 
\end{center}
\end{comment}

\subsection{The lower bound constructions}  
%Finally, we want to note that the lower-bound constructions of \cite{ms06},\cite{ss04},\cite{ss05} and \cite{hz16}
%are all based on the two lemmas above. The upper bounds for their niceness that we claim in Section~\ref{sec:nicenessknown} are 
%all based on the properties that we discuss above. 

Firstly, we consider Random Facet.
Matou\v{s}ek provided a family of LP-type problems that serve as a subexponential lower bound for 
Random Facet \cite{matouseklbrandomfacet}. Later, G\"{a}rtner translated these to AUSO \cite{G02} 
on which the algorithm requires a subexponential number of pivot steps that is asymptotically 
matching the upper bound in the exponent. The orientations that serve as a lower bound here are 
decomposable and, thus, 1-nice; RE can solve them in $O(n^2)$ steps. 
This is the only AUSO construction that we discuss in this section and is not based on the lemmas above.

Subsequently, consider the Bottom Antipodal algorithm. %This is a deterministic algorithm  \AT{We have described BA also in the previous section. We can cut down here.}
%which, at each step, jumps to the vertex that is antipodal in the cube spanned by the 
%out-edges of the current vertex. 
No non-trivial upper bounds are known for this 
algorithm. However, Schurr and Szab\'o \cite{ss05} have described AUSO on which Bottom Antipodal
takes $\Omega(\sqrt{2}^n)$ steps. Their constructions are 2-nice
 and, thus, RE can solve them with $O(n^3)$ steps. 

Last but not least, we discuss the lower bound constructions for RE. The first superpolynomial 
lower bound for RE on AUSO was proved by Matou\v{s}ek and Szab\'o in \cite{ms06}. 
Their construction achieves the lower bound
of $2^{\Omega{(n^{1/3})}}$ when Random Edge is started at a random vertex.
The construction is $n^{1/3}$-nice, which implies an upper bound of 
$2^{O(n^{1/3} \log n)}$ which is close to the lower bound.

Very recently, Hansen and Zwick \cite{hz16} improved these lower bounds by improving the 
techniques of \cite{ms06}. They achieve a lower bound of $2^{O(\sqrt{n \log n})}$.
Their construction is $\sqrt{n}$-nice; hence, we have an upper bound of 
$2^{O(\sqrt{n} \log n)}$ which is almost tight.

The upper bounds for the constructions of \cite{ss05}, \cite{ms06} and \cite{hz16} are all based
on Lemmas~\ref{lem:product} and~\ref{lem:hypersink}. The niceness bounds we mention follow 
directly from our arguments on the preservation of niceness fo these lemmas.
We turn our attention to cyclic USO.

For cyclic USO we have the lower bound provided
by Morris in \cite{morris}. The number of steps required when RE starts 
at any vertex of a Morris USO is at least $\frac{n-1}{2}! = n^{\Omega(n)}$ which is significantly 
larger than the number of all vertices. 
The lower bound implies that Morris USO is $i$-nice for $i = \Omega(n)$; 
otherwise, the upper bounds of Theorem~\ref{thm:UB} would contradict the
 lower bound. Indeed, as we explain in Section~\ref{sec:lbUSO}, Morris USO are 
 exactly $n$-nice and thus the upper bound we get by Theorem~\ref{thm:UB} is 
 tight to the lower bound.

In conclusion, using the niceness concept we can argue, firstly, that RE can solve
instances that serve as lower bounds for other algorithms in polynomial time. Secondly,
that on the lower bounds instances for RE the upper bounds are tight or almost tight.
We summarize the findings of this section in the following table: % (which also appears in Section~\ref{sec:REnice}):
\begin{center}
\begin{tabular}{l|c|c|c|c} 
Algorithm & Reference & Lower bound & Niceness & RE Upper bound \\ \hline
Random Facet & \cite{matouseklbrandomfacet},\cite{G02} & $ 2^{\Theta(\sqrt{n})}$ & $1$ & $O(n^2)$ \\
%General deterministic& \cite{ss04} & $\Omega(n^2 / \lceil \log n \rceil)$ & $\lceil \log n \rceil$ & $n^{O(\lceil \log n \rceil)}$  \\
Bottom Antipodal & \cite{ss05} & $\Omega(\sqrt{2}^n)$ & 2 & $O(n^3)$ \\
RE acyclic &\cite{ms06} & $2^{\Omega{(n^{1/3})}}$ & $n^{1/3}$ & $2^{O(n^{1/3} \log n)}$ \\
RE acyclic & \cite{hz16} & $2^{\Omega(\sqrt{n \log n})}$ & $\sqrt{n}$ & $2^{O(\sqrt{n} \log n)}$ \\
RE cyclic & \cite{morris} & $\frac{n-1}{2}!$ & $n$ & $n^{O(n)}$  
\end{tabular} 
\end{center}

\section{Counting 1-nice USO} \label{sec:count}
%The class of decomposable USO is the largest class known to be polynomially solvable by Random Edge. 
%we can now argue that the class of 1-nice USO is much larger than the class of decomposable ones. 
We have mentioned that the class of decomposable USO are 1-nice in Section~\ref{sec:REnice}. 
This class is the  
previously known largest class of USO, where Random Edge is polynomial. 
In this section we prove that the number of 1-nice USO is strictly larger than the number of 
decomposable USO. To the best of our knowledge a counting argument for 
decomposable USO does no exist in the literature. Thus, we provide one 
in Theorem~\ref{thm:decomp} and prove that 
the number of decomposable USO is $2^{\Theta(2^n)}$.

\begin{theorem} \label{thm:decomp}
The number of decomposable USO is $2^{\Theta(2^n)}$.
\end{theorem}
\begin{proof}
Firstly, we analyze the following recurrence relation and then we explain how it is derived from counting the number of decomposable USO.
Let 
\[
F(n) = P(n)\cdot F(n-1)^2, \quad n>0,
\]
where $P$ is some positive function defined on the positive integers, and $F(0)$ is some fixed positive value. Taking (binary) logarithms, we equivalently obtain
\[
\log F(n) = \log P(n) + 2 \log F(n-1), \quad n>0.
\]
If we substitute $f(n) := \log F(n)$ and $p(n) := \log P(n)$ we arrive at
\[
f(n) = p(n) + 2 f(n-1), \quad n>0.
\]
Simply expanding this yields
\begin{eqnarray*}
f(n) &=& \sum_{i=0}^{n-1}2^i p(n-i) + 2^n f(0) \\
      &=& \sum_{i=1}^{n}2^{n-i} p(i) + 2^n f(0) \\
      &=& 2^n\left(f(0) + \sum_{i=1}^{n}2^{-i} p(i)\right).  
\end{eqnarray*}

We conclude that $f(n) = \Theta(2^n)$, if the infinite series $\sum_{i=1}^{\infty}2^{-i} p(i)$ converges.
A sufficient condition for this is $p(n) \leq c^n$ for $c<2$, or $P(n) \leq 2^{c^n}$.

Now, let us explain how the above recurrence is derived. Let $F(n)$ denote 
the number of different decomposable USO. Consider the coordinate $n$. We can orient it in a combed way: all edges are forward
or all edges are backward. In the two antipodal facets defined by this coordinate we can embed any decomposable orientation. Thus, we 
have $F(n) \geq 2 \cdot F(n-1)^2$. 

The upper bound follows from the same procedure but now we allow to choose the combed coordinate at every step. Again,
for the coordinate we choose we have two choices (edges oriented forward or backward). Hence, we have 
$F(n) \leq 2n \cdot F(n-1)^2$. Note that the construction we suggest, i.e. taking any two $(n-1)$-USO and connecting them with 
a combed coordinate to an $n$-USO is safe by Lemma~\ref{lem:product}.

In conclusion, we have that $2 \leq P(n) \leq 2n$ and, thus, the infinite series we discuss above converges and $f(n) = \Theta(2^n)$.
It follows that $F(n) = 2^{\Theta(2^n)}$.
\end{proof}

We can now argue that the class of 1-nice USO is much larger than the class of decomposable ones, 
and also contains cyclic USO.
Actually, we can give a lower bound of the form $n^{c 2^n}$, for some 
constant $c$. 
To achieve this lower bound, we use the same technique that Matou\v{s}ek \cite{matousekcount} used 
to give a lower bound on the number of all USO, which is by counting \emph{flip-matching} orientations (FMO).

Consider any uniform orientation, i.e. all edges are oriented from the global source to the global sink. 
Pick any matching of the edges 
and reverse the orientation of those edges. The result is an FMO. It is known that FMO are USO \cite{matousekcount, ss04} 
(this can be proved as a corollary of Lemma~\ref{lem:hypersink}, Corollary 6 in \cite{ss04}).
Note that an FMO can be cyclic: the cyclic USO in Figure~\ref{fig:cubes}c is an FMO. It can be obtained 
by starting from the backward-uniform orientation, i.e. all edges are backward, and flip the 3 edges 
that appear forward in the figure. 

\newpage
\begin{theorem} \label{thm:1niceLB}
The number of 1-nice $n$-dimensional USO is $n^{\Theta(2^n)}$.
\end{theorem}
\begin{proof}
	Consider the following inductive construction. % that is a generalization of the Klee-Minty cube. 
	Let $A_1$ be any 1-dimensional USO. 
	Then, we construct $A_2$ by taking any 1-dimensional USO $A'_1$ and directing all edges on coordinate 2 towards $A_1$. 
	In general, to construct $A_{k+1}$: we take $A_k$ and put antipodally any $k$-dimensional USO $A'_k$. 
	Then, we direct all edges on coordinate $(k+1)$ 
	towards $A_k$. This is safe by Lemma~\ref{lem:product}.
	%Let us call this construction ``target-combed'' as there exists a path to the global sink, from any vertex, with  
	This construction satisfies the following property: 
	%the following property:
	for every vertex, the minimal face that contains this vertex and the global sink has a combed coordinate.
	We call such a USO \emph{target-combed}. It constitutes a generalization of decomposable USO.
	An illustration appears in Figure~\ref{fig:tcombed}. 
	\begin{figure}[htbp]
		\centering
		\includegraphics[scale=0.8]{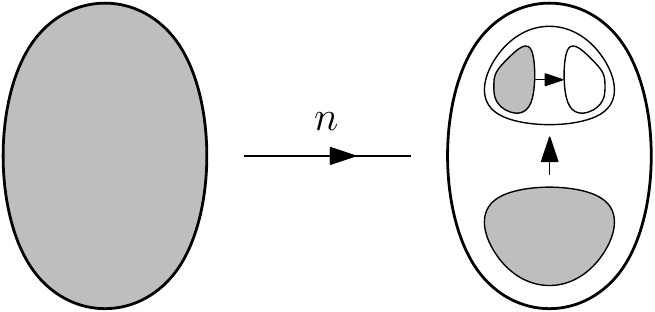}
		\caption{A target-combed $n$-USO. The two larger ellipsoids represent the two antipodal facets $A_{n-1}$
			and $A'_{n-1}$ and,  similarly, for the smaller ones. The combed coordinates are highlighted. The gray subcubes 
			can be oriented by any USO.} 
		\label{fig:tcombed}
	\end{figure}
	
	The construction is 1-nice since for every vertex (except the sink) there is an outgoing coordinate that can never be reached again.
	%Then, use the lower bounds from the proof of Matou\v{s}ek \cite{matousekcount}. 
	At every iteration step from $k$ to $k+1$ we can embed, in 
	one of the two antipodal $k$-faces, any USO.
	Thus, we can use the lower bounds of 
	\cite{matousekcount}, that give us a $\left( \frac{k}{e} \right)^{2^{k-1}}$ (assuming $k\geq 2$)
	lower bound for a $k$-face. This lower bound follows from counting different FMO by using a lower bound  
	on the number of perfect matchings of the hypercube graph.
	Summing up, we get: % the following lower bound:
	\[ uso1nice(n) \geq \sum_{k=1}^{n-1} uso(k) > uso(n-1) = \left( \frac{n-1}{e} \right)^{2^{n-2}} \]  
	where $uso1nice(n)$ and $uso(n)$ is the number of $n$-dimensional 1-nice USO and general USO respectively. 
	Thus, $uso1nice(n) = n^{\Omega(2^n)}$.
	The upper bound in the statement of the theorem is from the upper bound on the number of all USO, 
	by Matou\v{s}ek \cite{matousekcount}.
\end{proof}

\section{Bounds on niceness} \label{sec:bounds}  
In this section we answer the questions originally posed by Welzl \cite{emonice} with providing matching 
upper and lower bounds on the niceness of USO and AUSO. 

The first question we deal with here is ``Is there a USO that is not $(n-1)$-nice?''. 
We answer this to the affirmative; thus, the corresponding USO are only $n$-nice. They are,
however, cyclic. After we settle this we turn our attention to AUSO and prove matching 
upper and lower bounds: every AUSO is $(n-2)$-nice and there are AUSO that are not
$(n-3)$-nice.

\subsection{An $n$-nice lower bound for cyclic USO} \label{sec:lbUSO}
In this section we will provide a lower bound for the niceness of cyclic USO. The result is summarized in Theorem~\ref{thm:usolb}. 
First, consider the following lemma which follows easily from the definition of reachmap. 

\begin{lemma} \label{lem:fdimcyc}
Consider an n-dimensional USO $\psi$ and let $C \subset Q_n$ be a cycle that spans every coordinate. Then, 
every vertex $v \in C$ has $r_\psi(v) = [n]$.
\end{lemma}

The idea for the lower bound construction is intuitively simple and follows from the lemma above.
 Let $\psi$ be a cyclic $n$-USO 
over $Q^n$ that contains a directed cycle such that the edges that participate 
span all the coordinates. Then, every vertex $v$ on the cycle has $r_\psi(v) = [n]$.
Now consider the sink $t$ and assume the $n$ vertices in $\nd(t)$ participate in the 
cycle. By Lemma~\ref{lem:pathtosink}, every vertex has a path to $t$. This path has 
to go through one of the vertices in $\nd(t)$. It follows that every $v \in Q^n \setminus \{t\}$ has
$r_\psi(v)=[n]$. Therefore, the vertex antipodal from $t$ is only $n$-covered (by $t$).
This intuition is formalized in Theorem~\ref{thm:usolb}.

Note that the properties we just described are also satisfied by the Morris USO. This can be verified 
easily; for the interested reader we suggest these lecture notes \cite{RA} where there is a description of the 
Morris construction as a USO (the paper by Morris \cite{morris} describes it as a P-LCP). Thus,
Morris USO are $n$-nice.

Here we describe a much simpler USO, which is also an FMO
(a definition of FMO can be found in Section~\ref{sec:count}).  
The reason we think this is interesting is because it demonstrates that there are USO with large niceness, as in the example below, 
which RE can solve fast. It should be clear by the construction below that RE can solve it with polynomially many steps. In contrast, 
for Morris USO it will take more steps than the number of vertices of the hypercube.

\begin{theorem}\label{thm:usolb}
There exists a cyclic USO $\psi$ which is not $i$-nice for $i < n$. 
\end{theorem}
\begin{proof}
We describe a family of cyclic FMO that contain a cycle which spans all the coordinates. 
To achieve this we start with the forward uniform orientation $\psiu$ and
flip $n$ edges to create a cycle $C$ with $2n$ vertices. As explained in Section~\ref{sec:bounds},
the trick is to involve all the vertices in $Q_n^{n-1}$ in the cycle. Consider the vertex  $[n] \setminus \{n\} \in Q_n^{n-1}$.
We can construct the desired cycle as follows: 
\begin{figure}[htbp]
	\centering
	\includegraphics[width=\textwidth]{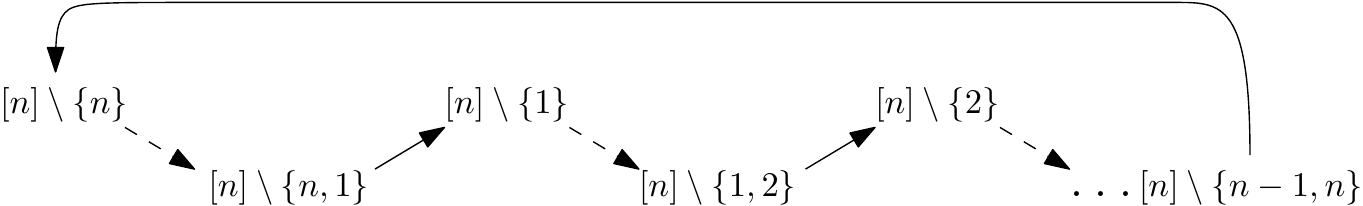}
	\caption{An illustration of the cycle. The dashed edges are flipped backwards.} 
	\label{fig:cyclelb}
\end{figure}

Of course, we can create this cycle by flipping exactly $n$ edges, one in each coordinate. 
This concludes the construction of $\psi$, our target USO.
Thus, $C$ spans all 
coordinates and by the lemma above, every vertex $v \in C$, has $r_\psi(v) = [n]$.

\begin{claim}
We have $r_\psi(v) = [n]$, $\forall v \in Q_n \setminus [n]$.
\end{claim}
We have that $\psi$ is an FMO and all backward edges are incident to vertices in $Q_n^{k}$, for $k=n-2,n-1$.
Thus, we have that all vertices in  $\bigcup_{k=0}^{n-3} Q_n^{k}$ are only incident to forward edges. 
Every vertex in $Q_n^{n-2}$ has at least one forward edge to a vertex in $Q_n^{n-1}$. 
We already argued that all vertices in $Q_n^{n-1}$ have full dimensional reachmaps. 
Therefore, every vertex in $\bigcup_{i=0}^{n-2} Q_n^{i}$ has a full dimensional reachmap too.

By the claim above, it follows that the only vertex that $i$-covers any other vertex $v$ is $[n]$. In particular, we have that the vertex $\emptyset$
is only $n$-covered. This concludes the proof of the theorem.
\end{proof}

%An example instance of the construction for 4 dimensions can be found in Figure~\ref{fig:4d2ncycle} below. 
Note that the 3-dimensional cyclic USO (depicted in Figure~\ref{fig:cubes}c) is also an instance of the construction suggested in this section. 

\subsection{An upper bound for AUSO}
Here, we prove an upper bound on the niceness of AUSO which, as we will see in the next section, is tight. 

We utilize the concept of Completely Unimodal Numberings (CUN), which was studied
by Williamson-Hoke \cite{wh} and Hammer \etal \cite{hammer}. To the best of our knowledge, this 
is the first time CUN is used to prove structural results for AUSO.
A CUN on the hypercube $Q^n$ means that
there is a bijective
function $\phi: Q^n \rightarrow \{0, \ldots, 2^n-1\}$ such that in every face $F$
there is exactly one vertex $v$ such that $\phi(v) < \phi(u)$, for every $u \in \nd(v) \cap F$. 
It is known, e.g. from \cite{wh}, that for every AUSO there is a corresponding CUN,
which can be constructed by topologically sorting the AUSO.

In the proof of the theorem below we will use the following notation: $w^k$ is the vertex 
that has $\phi(w^k) = k$, w.r.t. some fixed CUN $\phi$. An easy, but crucial observation concerns the three 
lowest-numbered vertices $w^0,w^1,w^2$. Of course, $w^1 \rightarrow w^0$ (where $w^0$ is the global sink); 
otherwise, $w^1$ would have been a second global minimum. Moreover, $w^2 \rightarrow w^j$
for exactly one $j \in \{0,1\}$. It follows, that both $w^1$ and $w^2$ are facet sinks.
We are ready to state and prove the following theorem.

\begin{theorem} \label{thm:UBnice}
Any $n$-AUSO, with $n \geq 4$, is $(n-2)$-nice.
\end{theorem}

Consider the vertices $w^0$ and $w^1$ and let $e$ be the edge that connects them. 
Let $w\in e$ be the unique out-neighbor of $w^2$ and $w'$ the other vertex in
$e$. W.l.o.g.\ assume $w=\emptyset, w'=\{1\}$ and  $w^2=\{2\}$.
The situation can be depicted as:
\begin{center}
\vspace{-0.1cm}
\includegraphics[scale=1]{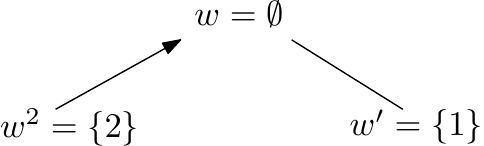}
\end{center}

These three vertices have no outgoing edges to other vertices. 
Their outmaps and reachmaps are summarized in the table below. 
\vspace{-0.1cm}
\begin{center}
\begin{tabular}{l|l|l|l} 
vertex & outmap & reachmap & is sink of the facet \\ \hline
$w = \emptyset$ & $\subseteq \{1\}$ & $\subseteq \{1\}$ & $F_{[n]\setminus\{1\},w}$ \\
$w' = \{1\}$ & $\subseteq \{1\}$ & $\subseteq \{1\}$ & $F_{[n]\setminus\{1\},w'}$\\
$w^2 = \{2\}$ & $=\{2\}$ & $\subseteq \{1,2\}$ & $F_{[n]\setminus\{2\},w^2}$
\end{tabular}
\end{center}
%\vspace{-0.1cm}
More precisely, the reachmap of $w^2$ is $\{2\}$ if $w=w^0$, and it is
$\{1,2\}$ if $w=w^1$. %From this table, we immediately get the following.

%\begin{lemma} With $w,w'$ as above,
%\begin{itemize}
%\item[(i)] $w=\emptyset$ is the sink of the facet $F_{[n]\setminus\{1\},w}$. 
%\item[(ii)] $w'=\{1\}$ is the sink of the facet $F_{[n]\setminus\{1\},w'}$. 
%\item[(iii)] $w^2=\{2\}$ is the sink of the face $F_{[n]\setminus\{1,2\},w^2}$.
%\end{itemize}
%\end{lemma}

\begin{lemma}
  With $w,w'$ as above, let $v \in Q^n \setminus \{w^0,[n]\}$. Then $v$ is $(n-2)$-covered by
some vertex in  $\{w,w',w^2\}$.
\end{lemma}
\begin{proof}
Vertex $w^1$ is covered by $w^0$ and $w^2$ by $w^0$ or $w^1$,
so assume that $v$ is some other vertex.

 If $v$ neither contains $1$
nor $2$, then $v$ is in the facet $F_{[n]\setminus\{1\},w}$. Hence,
$d(v,w) = |v\oplus w|\leq n-2$. This is because $F_{[n]\setminus\{1\},w}$
is $(n-1)$-dimensional and $2 \notin v$.
Any coordinate that is part of the corresponding path is in
the reachmap of $v$ but not of $w$ (whose reachmap is a subset of
$\{1\}$). Hence, $v$ is $(n-2)$-covered by $w$.

If $v$ contains $1$, then $v$ is in the facet
$F_{[n]\setminus\{1\},w'}$, and $|v\oplus w'|\leq n-2$ since $v\neq[n]$.
As before, this implies that $v$ is $(n-2)$-covered by the sink $w '$
of the facet in question.

Finally, if $v$ contains $2$ but not $1$, then $v$ is in the face
$F_{[n]\setminus\{1,2\},w^2}$, and $d(v,w^2)\leq n-2$. Again, any
coordinate on a directed path from $v$ to $w^2$ within this face proves
that $v$ is $(n-2)$-covered by the sink $w^2$ of the face. 
\end{proof}

It remains to $(n-2)$-cover the vertex $v=[n]$. Let $m>2$ be the smallest 
index such that $w^m$ is not a neighbor of $w$, and assume
w.l.o.g. that $w^k=\{k\}, 3\leq k<m$. We have $w^k\rightarrow w$ for
all these $k$ by the vertex ordering. Furthermore, all other edges  
incident to $w^k$ are incoming. We conclude that each 
$w^k, 3\leq k<m$ has outmap equal to $\{k\}$
and, hence, is a facet sink. The reachmap of each such $w^k$ is $\subseteq \{1,k\}$. 
The situation is depicted as:
%For the same reason, all edges
%to other adjacent vertices of the $w_i$ are incoming, meaning that
%each $w_i, 3\leq i<m$ has both outmap and reachmap equal to $\{i\}$
%and hence is a facet sink.
\begin{center}
\includegraphics[scale=1]{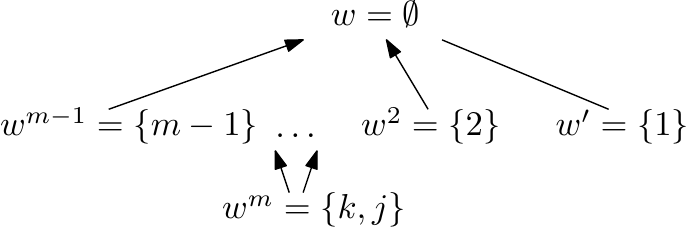}
\end{center}
%\vspace{-0.2cm}
Since $w^m$ has at least one out-neighbor in $\{w',w^2,\ldots,w^{m-1}\}$,
we know that $w^m=\{k,j\}$ for some $k<j\in[n]$. Moreover, the vertex
ordering again implies that the outgoing edges of $w^m$ are exactly
the ones to its (at most two) neighbors among $w',w^2,\ldots,w^{m-1}$.
Taking their reachmaps into account, we conclude that the reachmap of
$w^m$ is $\subseteq \{k,j,1\}$. 
%where the extra $1$ comes in if $2\in\{k,j\}$ and $w^2$ has reachmap $\{1,2\}$.

\begin{lemma}
With $w^m$ as above and $n\geq 4$, $v=[n]$ is $(n-2)$-covered by $w^m$.
\end{lemma}
\begin{proof}
We first observe that $w^m$ is the sink of the face
$F_{[n]\setminus\{k,j\}, w^m}$, since its outmap is $\subseteq \{k,j\}$.
Vertex $v=[n]$ is contained in this $(n-2)$-face, hence
there exists a directed path of length $d(v,w^m)=n-2$ from $v$ to
$w^m$ in this face. Since $n\geq 4$, the path spans at least two coordinates and
thus at least one of them is different from $1$. This coordinate proves that $v$ is
$(n-2)$-covered by $w^m$.
\end{proof}
To sum up, we have now proved that every $n$-AUSO, with $n \geq 4$,
is $(n-2)$-nice. All AUSO in one or two dimensions are 1-nice and the AUSO 
in three dimensions can be up to 2-nice (Figure~\ref{fig:cubes}b).
This concludes the upper bounds on the niceness of AUSO. % the proof of Theorem~\ref{thm:UBnice}.

\subsection{A matching lower bound for AUSO} \label{sec:lbAUSO}
In this section we prove the lower bound described in Theorem~\ref{thm:lbAUSO}, which is 
matching to the upper bound we proved in Theorem~\ref{thm:UBnice}. It is true that 
in a USO we can \emph{flip} any edge if the outmaps of the two vertices incident to it are the same 
(except the connecting coordinate) and still have a USO (Corollary 6, \cite{ss04}). 
We make use of this in the proof of the following theorem.

%We prove a lower bound on the niceness of acyclic USO that matches     
%the upper bound of Theorem~\ref{thm:UBnice}. It follows (Corollary 6, \cite{ss04})
%from the Hypersink Reorientation Lemma (Appendix~\ref{app:nicenessknown}) that in a USO we can \emph{flip} any edge 
%if the outmaps of the two vertices incident to it are the same (except the connecting coordinate).
%This gives rise to a particular family of USO, the Flip-Matching Orientations (FMO): those arise   
%when we start with a uniform orientation, e.g. all edges are forward, and we flip the edges 
%of an arbitrary matching. FMO have been studied in \cite{ss04} and \cite{ms06}.

\begin{theorem} \label{thm:lbAUSO}
There exists an $n$-AUSO $\psi$ which is not $i$-nice, for $i < n-2$. 
\end{theorem}
\begin{proof}
Let $\psiu$ be the forward uniform orientation, i.e. the orientation where all edges are forward.  
We explain how to construct $\psi$, our target orientation, starting from $\psiu$. With $Q^n_k$ we denote 
the set of vertices that contain $k$ coordinates, i.e. $|Q^n_k| = \binom{n}{k}$. We assume $n\geq 4$.
The idea here is to construct an AUSO that has its source at $\emptyset$ and has the property 
that every vertex in $\bigcup_{i=0}^{n-3} Q^n_{i}$ has a full-dimensional reachmap.

Pick $v\in Q^n_{n-3}$ and assume w.l.o.g. that $v=[n]\setminus \{1,2,3\}$.  
Consider the 2-dimensional face $F_{\{1,2\}, v}$ and direct the edges in this face backwards. 
This is the first step of the construction and it results in $s_\psi(v) = \{3\}$. 

For the second step, consider the vertex $v' = [n] \setminus \{2\}$. We will flip $n-3$ edges in order to create
a path starting at $v'$.
%The figure below suggests which edges to flip: 
First, we flip edge $F_{\{4\}, [n]\setminus \{2\}}$. 
Then, for all $k \in \{4,\ldots, n-1\}$ we flip the edge $F_{\{k+1\}, [n]\setminus \{k\}}$.
This creates the path depicted in Figure~\ref{fig:ausolb}. 
\begin{figure}[htbp]
	\centering
	\includegraphics[width=\textwidth]{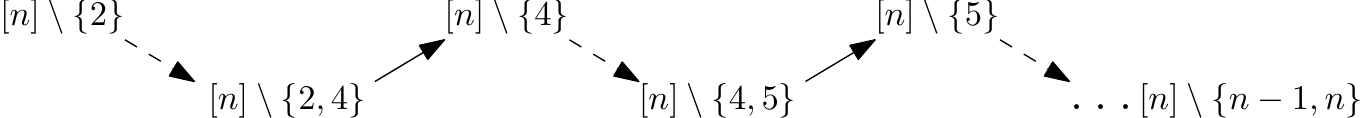}
	\caption{An illustration of the path starting at $v'$. The dashed edges are flipped backwards.} 
	\label{fig:ausolb}
\end{figure}

Let $U_3$ be the set of vertices $U_3 = \{u \in Q^n_{n-3} | 3 \in u \}$. 
That is all the vertices of $Q^n_{n-3}$ that contain the 3rd coordinate. 
For every $u \in U_3$ we flip the edge $F_{\{3\}, u}$ (that is the edge incident to $u$ on the 3rd coordinate). 
This is the third and last step of the construction of $\psi$.

\begin{claim} \label{cl:isuso}
$\psi$ is a USO.
\end{claim}
The first step of the construction is to flip the four edges in $F_{\{1,2\}, v}$. This is safe by considering that we first flip the two edges 
on coordinate 1; then, it is also safe to flip the two edges on coordinate 2.
All the edges reversed at the second step of the construction (Figure~\ref{fig:ausolb}) are between vertices in $Q^n_{n-1}$ and $Q^n_{n-2}$, and, in addition
those vertices are not neighbors to each other. Furthermore, all the edges reversed at the third step of the construction are on coordinate 3 and 
between vertices in $Q^n_{n-3}$ and $Q^n_{n-4}$. Thus, all these edge flips are safe. Note however that edge flips do not necessarily maintain
acyclicity (e.g. the cyclic USO in Figure~\ref{fig:cubes}c is an FMO); 
%All the edges we flip are safe by Corollary 6 in \cite{ss04}. Note, however that Corollary 6 does not guarantee 
%something about the creation of cycles; 
we have to verify acyclicity in a different way. 

\begin{claim} \label{cl:nocycle}
There is no cycle in $\psi$. 
\end{claim}
Clearly, a cycle has at least one backward and one forward edge in every coordinate it contains.
Thus, 
there cannot be a cycle that involves coordinate 3 because no backward edge on a 
different coordinate, has a path connecting it to a backward edge on coordinate 3.
%Clearly, a cycle has at least one backward and one forward edge in every coordinate it contains.

Consider the facet $F_{[n]\setminus \{3\}, [n]}$ and the USO $\psi'$, resulting from restricting $\psi$ to the aforementioned facet. We can notice that $\psi'$ is an FMO
and the only backward edges are the ones attached to the path illustrated in Figure~\ref{fig:ausolb}. Thus, a cycle has to use a part of this path. However, 
this path cannot be part of any cycles: a vertex on the higher level (vertices in $Q^n_{n-1}$) of the path has only two outgoing edges; one to the sink $[n]$ and one 
to the next vertex on the path. A vertex on the lower level $Q^n_{n-2}$ has only one outgoing edge to the next vertex on the path. Also, the last vertex of the
path $[n] \setminus  \{n-1,n\}$ has only one outgoing edge to $[n]\setminus \{n\}$ which has only one outgoing edge to the sink $[n]$.
%In the series of vertices 
%illustrated in Figure~\ref{fig:ausolb}, there is no repeated vertex. The last vertex in the Figure is $v'' = [n]\setminus \{n-1,n\}$, which has two edges incident to higher %vertices, namely $F_{\{n-1\}, v''}$ and $F_{\{n\}, v''}$. However,
%$F_{\{n\}, v''}$ is flipped backwards and thus the only forward out-edge incident to $v''$ is $F_{\{n-1\}, v''}$
%($v''$ is the sink of facet $F_{[n]\setminus \{n-1\}, v''}$). 
%We have $v'' \xrightarrow{n-1} [n] \setminus {\{n\}}$. 
%The vertex $[n] \setminus {\{n\}}$ has no backwards edge incident to it and thus there is no cycle in $\psi'$. 

The fact that the facet $F_{[n]\setminus \{3\}, v}$ has no cycle follows from the observation that there are backwards edges only on two coordinates
which is not enough for the creation of a cycle (remember that in a USO a cycle needs to span at least three coordinates).
This concludes the proof of Claim~\ref{cl:nocycle}, which, 
combined with Claim~\ref{cl:isuso}, results in $\psi$ being an AUSO. 

\begin{claim}
Every vertex in $\bigcup_{i=0}^{n-3} Q^n_{i}$ has a full-dimensional reachmap.
\end{claim}
Firstly, we argue that $v$ has $r_\psi(v)=[n]$. We have $s_\psi(v) = \{3\} \subset r_\psi(v)$. Then, $v \xrightarrow{3} u = [n] \setminus \{1,2\}$ and $u$ has 
$s_\psi(u) = \{1,2\} \subset r_\psi(v)$. Vertex $u$ is such that $u \xrightarrow{1} v' = [n]\setminus \{2\}$; $v'$ is the 
beginning of the path described in Figure~\ref{fig:ausolb}. The backwards edges on this path
span every coordinate in $\{4,\ldots, n\}$. This implies that $r_\psi(v') = \{2,4,\ldots, n\}$ and, since there is a path from $v$ to $v'$,
$r_\psi(v') \subseteq r_\psi(v)$. Combined with the above, we have that $r_\psi(v)=[n]$.

Secondly, we argue that $\forall u \in Q^n_{n-3}$, $r_\psi(u)=[n]$. Vertex $v$ is the sink of the facet $F_{[n]\setminus \{3\}, v}$. 
It follows that every vertex in $Q^n_{n-3} \cap F_{[n]\setminus \{3\}, v}$ has a path to $v$ and thus has full dimensional reachmap.
The vertices in $U_3$ (defined earlier), which are the rest of the vertices in $Q^n_{n-3}$, have backward edges on coordinate 3 and thus 
have paths to $F_{[n]\setminus \{3\}, v}$. It follows that vertices in $U_3$ also have full dimensional reachmaps.

Any vertex in $\bigcup_{i=0}^{n-4} Q^n_{i}$ has a path to a vertex in $Q^n_{n-3}$ since there are outgoing forward edges 
incident to any vertex in $\psi$ (except the global sink at $[n]$). 
Thus, we have that $\forall u \in \bigcup_{i=0}^{n-3} Q^n_{i}$, $r_\psi(u) = [n]$ which proves the claim. 

Finally, we combine the three Claims to conclude that the lowest vertex $\emptyset$ can 
only be covered by a vertex in $Q^n_{n-2}$. Therefore, $\psi$ is \emph{not} $i$-nice 
for any $i < n-2$, which proves the theorem.
We include an example construction, for five dimensions, in Figure~\ref{fig:ausolbex} below.
\end{proof}
\vspace{-0.3cm}
\begin{figure}[htbp]
	\centering
	\includegraphics[scale=1]{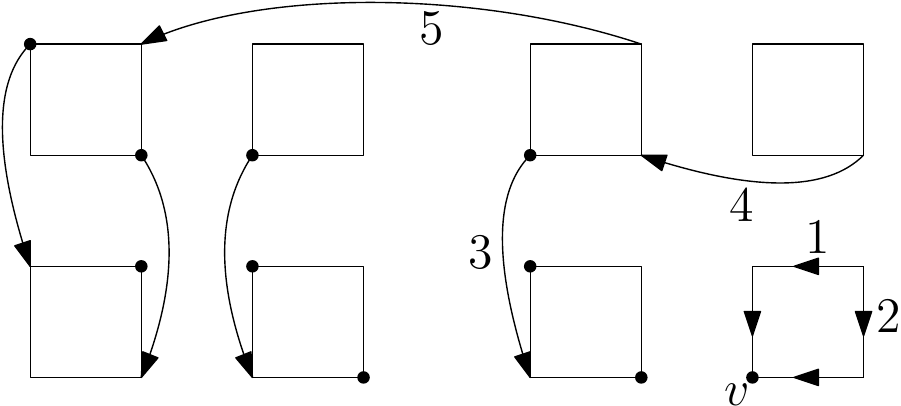}
	\caption{An example construction in 5 dimensions. Only the backward edges are noted. 
		Each coordinate is labeled over a backward edge.
		The 5-dimensional cube 
	is broken in 2-faces of coordinates 1,2. All the vertices in $Q^n_{n-3}$ are noted with dots. Also, $v$ is explicitly noted.} 
	\label{fig:ausolbex}
\end{figure}

\section{Fibonacci Seesaw revisited} \label{sec:FS}

In this section we motivate further the concept of reachmap with one more application. 
We introduce a variant of the Fibonacci Seesaw algorithm for solving USO 
(originally introduced in \cite{SW}). This new variant is interesting because the number of oracle calls it needs can be bounded 
by a function that is exponential to the size of the reachmap of the starting vertex, see Theorem~\ref{thm:FSv}.

Let $\psi$ be a USO. The Fibonacci Seesaw (FS) algorithm progresses by increasing a variable $j$ from 0 to $n-1$ while it maintains the following invariant:
There are two antipodal $j$-faces $A$ and $B$ of $Q^n$ that have their sinks $s_A$ and $s_B$ evaluated. For $j=0$ this means to evaluate two antipodal 
vertices. To go from $j=k$ to $j=k+1$ we take a coordinate $b \in s_\psi(s_A) \oplus s_\psi(s_B)$. Such a coordinate has to exist because 
of Lemma~\ref{lem:outmap}. Let $b \in s_\psi(s_A)$ and $b\notin s_\psi(s_B)$. Let $A'$ be the (k+1)-face that we get by extending $A$ with coordinate $b$ and $B'$ be the 
corresponding face from $B$.
We have that $s_B$ is the sink of $B'$. For $A'$ we need to evaluate the sink. But this will lie in the k-face $A'\setminus A$. Thus, for this step 
we need $t(k)$ evaluations, where $t(k)$ is the number of steps the FS needs to evaluate the sink of a $k$-USO.

When we reach two antipodal facets with $j=n-1$ the algorithm will terminate as either $s^j_A$ or $s^j_B$ will be the sink. The total cost of the algorithm 
is known to be $O(\alpha^n)$ vertex evaluations, where $\alpha < \phi$ is a constant slightly smaller than the golden ratio
(to have $\alpha$ strictly smaller than $\phi$ some further adjustments are needed; those can be found in \cite{SW}). 
Here, we consider the following algorithm:

\begin{algorithm}[H]
\caption{FS Revisited}
%\begin{algorithmic}
Set index $j=0$\;
%\State Pick a starting vertex $v^i \in Q^n$.
Pick a starting vertex $v^j \in Q^n$\;
Set evaluated coordinates $E^j = \emptyset$\;
\While{$s_\psi(v^j) \neq \emptyset$}{
	Pick $b \in s_\psi(v^j)$\;
	$v^{j+1} \leftarrow$ FS($F_{E^{j}, v^j \oplus \{b\} }$)\;
	$E^{j+1} \leftarrow E^{j} \cup \{b\}$\;
	$j \leftarrow j+1$\;
}
%\end{algorithmic}
\end{algorithm}
After the $j$th iteration, the above algorithm considers the sink $v^j$ of face $F_{E^{j}, v^j}$. Then, it expands the set of coordinates by adding a coordinate 
$b$ that is outgoing for $v^j$ and solving, using the Fibonacci Seesaw, the face $F_{E^{j}, v^j\oplus \{b\} }$. 
The result is the sink of the face
$F_{E^{j+1}, v^0}$ which becomes vertex $v^{j+1}$. The algorithm terminates when it has evaluated the global sink.

\begin{lemma}\label{lem:reachmap}
Let $\rho$ be the iteration in which Algorithm 1 terminates. Then,
\[
r_\psi(v^0) \supseteq r_\psi(v^1) \supseteq \ldots \supseteq r_\psi(v^\rho) \]
\end{lemma}
\begin{proof}
%Consider $s^0_A$ and w.l.o.g. assume that the first coordinate considered by the algorithm is $b\in s_\psi(s^0_A)$. Then 
Consider $v^j$ for any $j$ and let $j \in s_\psi(v^j)$ be the next coordinate that the algorithm will consider. 
In the next step, we will have $v^{j+1}$ which will be the sink of face $F_{E^{j}, v^j\oplus \{b\} }$, for some $b\in s_\psi(v^j)$. We have that 
 $F_{E^{j}, v^j\oplus \{b\} } \subset F_{E^{j+1}, v^j }$, for every $j \leq \rho$. In particular, we have 
 $F_{E^{j}, v^j\oplus \{b\} } \subset F_{E^{j+1}, v^0}$, for every $0 < j \leq \rho$. This meas that there is a path $v^0 \rightsquigarrow v^j$, for every 
$0 < j \leq \rho$. The lemma follows.
\end{proof}

Using the above lemma, we can upper-bound the number of iterations of Algorithm 1 in terms of the reachmap of the
starting vertex. 

\begin{lemma}\label{lem:iterations}
Let $\rho$ be the iteration in which the algorithm terminates. Then, $\rho \leq |r_\psi(v^0)|$.
\end{lemma}
\begin{proof}
After iteration $j$, Algorithm 1 has computed $v^j$ which is the sink of a $j$-face. 
We have argued that $r_\psi(v^0) \supseteq r_\psi(v^j)$ for every 
$0 < i \leq \rho$ at Lemma~\ref{lem:reachmap}. This means that the coordinate we pick at any iteration is in the reachmap of $v^0$. In addition,
the set of coordinates $E^j$ grows at every iteration. 
If $E^j = r_\psi(v^0)$ (or equivalently $j = |r_\psi(v^0)|$), then $v^j$ will be the sink of the face $F_{r_\psi(v^0), v^0}$. This means that 
$v^j$ will be the global sink of $\psi$, i.e. $s_\psi(v^j)=\emptyset$. Of course, it might happen 
that $v^j$ is the global sink for $E^j \subset r_\psi(v^0)$;
hence, the inequality. 
\end{proof}

\begin{theorem} \label{thm:FSv}
Algorithm 1, when run on an $n$-USO $\psi$ with starting vertex $v^0 \in Q^n$,
 needs at most $O(\alpha^\rho)$ vertex evaluations, where $\alpha < \phi$ and $\rho = |r_\psi(v^0)|$.
\end{theorem}
\begin{proof}
The algorithm performs at most $\rho$ iterations. 
At each of them it calls the Fibonacci Seesaw to solve a face of $\psi$. In particular, 
when progressing from the $j$th to the $(j+1)$th iteration it calls the Fibonacci Seesaw to solve an $j$-face of $\psi$.
Thus, the number of vertex oracle calls of Algorithm 1 can be bound by
\[
\sum_{k=0}^{\rho} \alpha^k = \frac{a^{\rho+1}-1}{\alpha-1} = O(\alpha^\rho)
\]
where $1 < \alpha < \phi$ is the constant in the time bounds of the Fibonacci Seesaw algorithm, i.e. $\alpha \approx 1.61$.
\end{proof}

We have that Algorithm 1 is faster asymptotically than the Fibonacci Seesaw 
when the size of the reachmap of the starting vertex is small,
i.e. $|r_\psi(v^0)| \ll n$. We believe that similar variants as the one above (adding one coordinate at a time and 
recursively running the algorithm in question) would provide upper bounds where the reachmap of the starting vertex
is in the exponent, for every non-path-following algorithm. 

%Of course, it might be the case that $|r_\psi(v^0)| = n$. Thus in the worst case, Algorithm 1 has the same asymptotic
%complexity as the Fibonacci Seesaw. Also, note that there are examples of USO (e.g. the cyclic $n$-nice USO is such an example)
%where all vertices but the sink have full-dimensional reachmaps. 

\section{Conclusions} 
In this paper we study the reachmaps and niceness of USO, concepts introduced by Welzl \cite{emonice} in 2001. 
The questions that Welzl originally posed are now answered and the concepts explored further. 
We believe that these tools, or related ones, will prove useful in finally closing the gap between the 
lower and upper bounds known for RE. This will happen with either exponential lower bounds or 
with subexponential upper bounds. It is worth mentioning that these concepts are not    
only relevant for USO, but could also be defined on generalizations of USO, such as Grid USO \cite{raey} or  
Unimodal Numberings \cite{hammer}. 

The authors of \cite{hz16} define the concept of
a $(k,\ell)$-layered AUSO and use it to argue that their 
lower bounds are optimal under the method they use. Their concept is a generalization of niceness (on AUSO)
but the exact relationship remains to be discovered. 
They pose the following questions: Are there AUSO that are not $(2^{O(\sqrt{n \log n)}}, O(\sqrt{n / \log n}))$-layered?
Are there small constants $c,d$ such that all AUSO are $(c^n , dn/ \log n)$-layered?
We believe that the techniques of our proofs 
from Theorems~\ref{thm:UBnice} and~\ref{thm:lbAUSO} may be fruitful for answering these questions.

%\subparagraph*{Acknowledgements}
%We would like to thank Thomas Dueholm Hansen and Uri Zwick for sharing their work \cite{hz16} with us.  

\bibliography{uso}

\begin{thebibliography}{10}

\bibitem{apr14}
Ilan Adler, Christos~H. Papadimitriou, and Aviad Rubinstein.
\newblock On simplex pivoting rules and complexity theory.
\newblock In Jon Lee and Jens Vygen, editors, {\em Integer Programming and
  Combinatorial Optimization - 17th International Conference, {IPCO} 2014,
  Bonn, Germany, June 23-25, 2014. Proceedings}, volume 8494 of {\em Lecture
  Notes in Computer Science}, pages 13--24. Springer, 2014.

\bibitem{aoshima}
Yoshikazu Aoshima, David Avis, Theresa Deering, Yoshitake Matsumoto, and Sonoko
  Moriyama.
\newblock On the existence of {H}amiltonian paths for history based pivot rules
  on acyclic unique sink orientations of hypercubes.
\newblock {\em Discrete Applied Mathematics}, 160(15):2104 -- 2115, 2012.

\bibitem{bp}
J\'{o}zsef Balogh and Robin Pemantle.
\newblock The {K}lee-{M}inty random edge chain moves with linear speed.
\newblock {\em Random Structures \& Algorithms}, 30(4):464--483, 2007.

\bibitem{ds15}
Yann Disser and Martin Skutella.
\newblock The simplex algorithm is {NP}-mighty.
\newblock In Piotr Indyk, editor, {\em Proceedings of the 26th Annual
  {ACM-SIAM} Symposium on Discrete Algorithms, {SODA} 2015, San Diego, CA, USA,
  January 4-6, 2015}, pages 858--872. {SIAM}, 2015.

\bibitem{fs15}
John Fearnley and Rahul Savani.
\newblock The complexity of the simplex method.
\newblock In Rocco~A. Servedio and Ronitt Rubinfeld, editors, {\em Proceedings
  of the 47th Annual {ACM} on Symposium on Theory of Computing, {STOC} 2015,
  Portland, OR, USA, June 14-17, 2015}, pages 201--208. {ACM}, 2015.

\bibitem{fs16}
John Fearnley and Rahul Savani.
\newblock The complexity of all-switches strategy improvement.
\newblock In Robert Krauthgamer, editor, {\em Proceedings of the 27th Annual
  {ACM-SIAM} Symposium on Discrete Algorithms, {SODA} 2016, Arlington, VA, USA,
  January 10-12, 2016}, pages 130--139. {SIAM}, 2016.

\bibitem{count}
Jan Foniok, Bernd G{\"a}rtner, Lorenz Klaus, and Markus Sprecher.
\newblock Counting unique-sink orientations.
\newblock {\em Discrete Applied Mathematics}, 163, Part 2:155 -- 164, 2014.

\bibitem{Friedmann11}
Oliver Friedmann.
\newblock A subexponential lower bound for {Z}adeh's pivoting rule for solving
  linear programs and games.
\newblock In Oktay G{\"{u}}nl{\"{u}}k and Gerhard~J. Woeginger, editors, {\em
  Integer Programming and Combinatoral Optimization - 15th International
  Conference, {IPCO} 2011, New York, NY, USA, June 15-17, 2011. Proceedings},
  volume 6655 of {\em Lecture Notes in Computer Science}, pages 192--206.
  Springer, 2011.

\bibitem{fhz11}
Oliver Friedmann, Thomas~Dueholm Hansen, and Uri Zwick.
\newblock Subexponential lower bounds for randomized pivoting rules for the
  simplex algorithm.
\newblock In Lance Fortnow and Salil~P. Vadhan, editors, {\em Proceedings of
  the 43rd Annual {ACM} Symposium on Theory of Computing, {STOC} 2011, San
  Jose, CA, USA, June 6-8, 2011}, pages 283--292. {ACM}, 2011.

\bibitem{G02}
Bernd G{\"{a}}rtner.
\newblock The random-facet simplex algorithm on combinatorial cubes.
\newblock {\em Random Structures \& Algorithms}, 20(3):353--381, 2002.

\bibitem{RA}
Bernd G\"{a}rtner.
\newblock {\em Randomized Algorithms: An Introduction through Unique Sink
  Orientations}.
\newblock Lecture notes, ETH Z\"urich, 2004.

\bibitem{ghz98}
Bernd G{\"{a}}rtner, Martin Henk, and G{\"{u}}nter~M. Ziegler.
\newblock Randomized simplex algorithms on {K}lee-{M}inty cubes.
\newblock {\em Combinatorica}, 18(3):349--372, 1998.

\bibitem{raey}
Bernd G\"{a}rtner, Walter~D. Jr.~Morris, and Leo R\"{u}st.
\newblock Unique sink orientations of grids.
\newblock {\em Algorithmica}, 51(2):200--235, 2008.

\bibitem{GS}
Bernd G{\"{a}}rtner and Ingo Schurr.
\newblock Linear programming and unique sink orientations.
\newblock In {\em Proceedings of the 17th Annual {ACM-SIAM} Symposium on
  Discrete Algorithms, {SODA} 2006, Miami, FL, USA, January 22-26, 2006}, pages
  749--757. {ACM} Press, 2006.

\bibitem{gt15}
Bernd G{\"{a}}rtner and Antonis Thomas.
\newblock The complexity of recognizing unique sink orientations.
\newblock In Ernst~W. Mayr and Nicolas Ollinger, editors, {\em 32nd
  International Symposium on Theoretical Aspects of Computer Science, {STACS}
  2015, March 4-7, 2015, Garching, Germany}, volume~30 of {\em LIPIcs}, pages
  341--353. Schloss Dagstuhl - Leibniz-Zentrum fuer Informatik, 2015.

\bibitem{hammer}
Peter~L. Hammer, Bruno Simeone, Thomas~M. Liebling, and Dominique de~Werra.
\newblock From linear separability to unimodality: {A} hierarchy of
  pseudo-boolean functions.
\newblock {\em {SIAM} J. Discrete Math.}, 1(2):174--184, 1988.

\bibitem{hpz14}
Thomas~Dueholm Hansen, Mike Paterson, and Uri Zwick.
\newblock Improved upper bounds for random-edge and random-jump on abstract
  cubes.
\newblock In Chandra Chekuri, editor, {\em Proceedings of the 25th Annual
  {ACM-SIAM} Symposium on Discrete Algorithms, {SODA} 2014, Portland, OR, USA,
  January 5-7, 2014}, pages 874--881. {SIAM}, 2014.

\bibitem{hz16}
Thomas~Dueholm Hansen and Uri Zwick.
\newblock Random-edge is slower than random-facet on abstract cubes.
\newblock In Ioannis Chatzigiannakis, Michael Mitzenmacher, Yuval Rabani, and
  Davide Sangiorgi, editors, {\em 43rd International Colloquium on Automata,
  Languages, and Programming, {ICALP} 2016, July 11-15, 2016, Rome, Italy},
  volume~55 of {\em LIPIcs}, pages 51:1--51:14. Schloss Dagstuhl -
  Leibniz-Zentrum fuer Informatik, 2016.

\bibitem{kalai92}
Gil Kalai.
\newblock A subexponential randomized simplex algorithm (extended abstract).
\newblock In S.~Rao Kosaraju, Mike Fellows, Avi Wigderson, and John~A. Ellis,
  editors, {\em Proceedings of the 24th Annual {ACM} Symposium on Theory of
  Computing, May 4-6, 1992, Victoria, BC, Canada}, pages 475--482. {ACM}, 1992.

\bibitem{klaus}
Lorenz Klaus and Hiroyuki Miyata.
\newblock Enumeration of {PLCP}-orientations of the 4-cube.
\newblock {\em European Journal of Combinatorics}, 50:138 -- 151, 2015.

\bibitem{kleeminty}
Victor Klee and George~J. Minty.
\newblock How good is the simplex algorithm?
\newblock {\em Inequalities III}, pages 159--175, 1972.

\bibitem{matouseklbrandomfacet}
Ji\v{r}\'{\i} Matou\v{s}ek.
\newblock Lower bounds for a subexponential optimization algorithm.
\newblock {\em Random Structures \& Algorithms}, 5(4):591--607, 1994.

\bibitem{matousekcount}
Ji\v{r}\'{\i} Matou\v{s}ek.
\newblock The number of unique-sink orientations of the hypercube*.
\newblock {\em Combinatorica}, 26(1):91--99, 2006.

\bibitem{msw96}
Ji\v{r}\'{\i} Matou\v{s}ek, Micha Sharir, and Emo Welzl.
\newblock A subexponential bound for linear programming.
\newblock {\em Algorithmica}, 16(4/5):498--516, 1996.

\bibitem{ms06}
Ji\v{r}\'{\i} Matou\v{s}ek and Tibor Szab\'o.
\newblock {RANDOM} {EDGE} can be exponential on abstract cubes.
\newblock {\em Advances in Mathematics}, 204(1):262 -- 277, 2006.

\bibitem{morris}
Walter~D. Morris~jr.
\newblock Randomized pivot algorithms for {P}-matrix linear complementarity
  problems.
\newblock {\em Mathematical Programming}, 92(2):285--296, 2002.

\bibitem{ss04}
Ingo Schurr and Tibor Szab{\'o}.
\newblock Finding the sink takes some time: An almost quadratic lower bound for
  finding the sink of unique sink oriented cubes.
\newblock {\em Discrete {\&} Computational Geometry}, 31(4):627--642, 2004.

\bibitem{ss05}
Ingo Schurr and Tibor Szab{\'{o}}.
\newblock Jumping doesn't help in abstract cubes.
\newblock In Michael J{\"{u}}nger and Volker Kaibel, editors, {\em Integer
  Programming and Combinatorial Optimization - 11th International Conference,
  {IPCO} 2005, Berlin, Germany, June 8-10, 2005. Proceedings}, volume 3509 of
  {\em LNCS}, pages 225--235. Springer, 2005.

\bibitem{StiWat}
Alan Stickney and Layne Watson.
\newblock Digraph models of {B}ard-type algorithms for the linear
  complementarity problem.
\newblock {\em Math. Oper. Res.}, 3(4):322--333, 1978.

\bibitem{SW}
Tibor Szab{\'{o}} and Emo Welzl.
\newblock Unique sink orientations of cubes.
\newblock In {\em 42nd Annual Symposium on Foundations of Computer Science,
  {FOCS} 2001, 14-17 October 2001, Las Vegas, NV, {USA}}, pages 547--555.
  {IEEE} Computer Society, 2001.

\bibitem{emonice}
Emo Welzl.
\newblock i-{N}iceness.
\newblock
  \url{http://www.ti.inf.ethz.ch/ew/workshops/01-lc/problems/node7.html}, 2001.

\bibitem{wh}
Kathy Williamson-Hoke.
\newblock Completely unimodal numberings of a simple polytope.
\newblock {\em Discrete Applied Mathematics}, 20(1):69--81, 1988.

\end{thebibliography}

\end{document}